\documentclass{article}

\usepackage{etex}
\usepackage{mleftright}

\usepackage{xparse}
\usepackage{sterling}
\usepackage{mathpartir}
\usepackage{stmaryrd}
\usepackage{tikz-cd}
\tikzcdset{every label/.append style = {font = \small}}

\newcommand\Parens[1]{\mleft(#1\mright)}
\newcommand\Squares[1]{\mleft[#1\mright]}

\newcommand\SetCompr[2]{\mleft\lbrace{#1}\,\middle\vert\,{#2}\mright\rbrace}
\newcommand\Var{\mathsf{v}}

\newcommand\Realize[2]{\FmtTm{#1}\mathrel{\normalcolor\Vdash}\FmtVal{#2}}

\newcommand\FmtNamedCat[1]{\mathfrak{#1}}
\newcommand\SIGN{\FmtNamedCat{Sig}}

\newcommand\SET{\FmtNamedCat{Set}}
\newcommand\ARR[1]{{#1}^{\to}}
\newcommand\Cod{\mathsf{cod}}
\newcommand\GlFib[1]{\mathsf{gl}_{#1}}
\newcommand\Ar[1]{{#1}^\dagger}

\NewDocumentCommand\Sorts{O{}}{%
  \mathscr{U}_{\Mute{#1}}%
}

\NewDocumentCommand\Ops{O{}}{%
  \partial_{\Mute{#1}}
}

\newcommand\Ty[1]{\widetilde{#1}}
\newcommand\Psh[1]{\widehat{#1}}
\newcommand\Sh[1]{\mathbf{Sh}\Parens{#1}}
\newcommand\Mute[1]{{\color{SlateGray}#1}}

\newcommand\FmtTm[1]{{\color{DarkSlateBlue}#1}}
\newcommand\FmtNF[1]{{\color{MediumOrchid}#1}}
\newcommand\FmtNE[1]{{\color{MediumOrchid}#1}}
\newcommand\FmtVal[1]{{\color{FireBrick}#1}}

\newcommand\Seq[4]{%
  {#2}\mathrel{\vdash_{\Mute{#1}}}\FmtTm{#3}:{#4}%
}

\newcommand\SeqEq[5]{%
  {#2}\mathrel{\vdash_{\Mute{#1}}}\FmtTm{#3}=\FmtTm{#4}:{#5}%
}

\newcommand\Clone[2]{\mathsf{Cn}_{\Mute{#1}}\Parens{#2}}
\newcommand\SbClone[2]{\mathsf{Sb}_{\Mute{#1}}\Parens{#2}}
\newcommand\ClCat[1]{\mathsf{Cl}_{\Mute{#1}}}

\NewDocumentCommand\Hom{O{}mm}{{#1}\Squares{#2,#3}}

\newcommand\Ren[1]{\mathsf{Ren}_{\Mute{#1}}}
\newcommand\Yo{\mathbf{y}}
\newcommand\FancyYo{[\mathbf{y}]}
\newcommand\TM{\mathfrak{Tm}}
\newcommand\SeqNE[4]{{#2}\mathrel{\vdash^{\mathsf{ne}}_{\Mute{#1}}}\FmtNE{#3}:#4}
\newcommand\SeqNF[4]{{#2}\mathrel{\vdash^{\mathsf{nf}}_{\Mute{#1}}}\FmtNF{#3}:#4}
\newcommand\NE[1]{\mathfrak{Ne}_{#1}}
\newcommand\NF[1]{\mathfrak{Nf}_{#1}}
\newcommand\VAR[1]{\mathcal{V}\Parens{#1}}
\NewDocumentCommand\SbEmp{O{}}{{[\,]}_{\Mute{#1}}}
\newcommand\GlCat[1]{\mathsf{Gl}_{#1}}
\newcommand\Sem[1]{\llbracket{}#1\rrbracket}
\newcommand\Pred[1]{\mathcal{R}_{#1}}

\NewDocumentCommand\Quo{mo}{%
  \FmtTm{\mathsf{quo}_{\normalcolor #1}}%
  \IfValueT{#2}{%
    \FmtVal{\Parens{#2}}%
  }%
}

\NewDocumentCommand\RbNF{O{}mo}{%
  \FmtTm{\mathsf{Rnf}^{\Mute{#1}}_{\normalcolor #2}}%
  \IfValueT{#3}{%
    \FmtNF{\Parens{#3}}%
  }%
}

\NewDocumentCommand\RbNE{O{}mo}{%
  \FmtTm{\mathsf{Rne}^{\Mute{#1}}_{\normalcolor #2}}%
  \IfValueT{#3}{%
    \FmtNE{\Parens{#3}}%
  }
}

\NewDocumentCommand\Reify{O{}mo}{%
  \FmtNF{\downarrow_{\Mute{#1}}^{\normalcolor #2}}%
  \IfValueT{#3}{%
    \FmtVal{\Parens{#3}}%
  }%
}

\NewDocumentCommand\Reflect{O{}mo}{%
  \FmtVal{\uparrow_{\Mute{#1}}^{\normalcolor #2}}%
  \IfValueT{#3}{%
    \FmtNE{\Parens{#3}}%
  }%
}

\NewDocumentCommand\NfFun{mmo}{%
  \FmtNF{%
    \mathbf{nf}_{\Mute{#1}}^{\normalcolor#2}%
  }%
  \IfValueT{#3}{%
    \FmtTm{\Parens{#3}}
  }%
}

\NewDocumentCommand\SbIdn{O{}m}{%
  \mathsf{id}^{\Mute{#1}}_{\normalcolor #2}%
}

\newcommand\SbProj{\mathsf{p}}
\newcommand\SbLift[1]{{\Uparrow}[#1]}
\newcommand\SbExt[2]{{#1}.{#2}}

\newcommand\ObjLam[2]{%
  \lambda^{\Mute{#1}}\Parens{#2}%
}

\newcommand\MetaLam[2]{
  {\normalcolor\boldsymbol{\lambda}}{#1}.\,{#2}
}

\title{
  Normalization by gluing for free $\lambda$-theories
}
\author{
  Jonathan Sterling\thanks{\texttt{jmsterli@cs.cmu.edu}}\\
  {\small Carnegie Mellon University}
  \and
  Bas Spitters\thanks{\texttt{spitters@cs.au.dk}}\\
  {\small Aarhus University}
}

\begin{document}
\maketitle

\begin{abstract}

  The connection between normalization by evaluation, logical predicates and
  semantic gluing constructions is a matter of folklore, worked out in varying
  degrees within the literature. In this note, we present an elementary version
  of the gluing technique which corresponds closely with both semantic
  normalization proofs and the syntactic normalization by evaluation

\end{abstract}

We will expand in more detail the insight presented in Streicher's short
note~\citep{streicher:1998} and in~\citet{fiore:2002}, giving some explicit
constructions. We will be considering the case of free $\lambda$-theories
generated from many-typed first-order signatures.

\section{\texorpdfstring{$\lambda$}{Lambda}-signatures and \texorpdfstring{$\lambda$}{lambda}-theories}

\begin{definition}[Arity]

  A simply-typed first-order arity for a set $T$ of atomic types is a pair
  $\alpha\equiv\Parens{\vec{\sigma},\tau}$ of a list of types and a types. We write
  $\Ar{T}\eqdef{}T^\star\times T$ for the set of such arities.

\end{definition}

\begin{definition}[Many-typed signature]

  Following~\citet{jacobs:1999}, a \emph{many-typed signature}
  $\Sigma\equiv\Parens{\Sorts,\Ops}$ is a set of atomic types $\Sorts$ together with an
  arity-indexed family of sets of operations $\Ops$, taking each operation $\vartheta$ to $\Ops(\vartheta)\in\Ar{\Sorts}$.  More abstractly, the category $\SIGN$ of such signatures
  arises as the pullback of the fundamental fibration along the arity
  endofunctor.
  \[
    \begin{tikzcd}[sep=huge]
      \SIGN
      \arrow[r, "\Ops"]
      \arrow[d, -{Triangle[open]}, swap, "\Sorts"]
      \arrow[dr, phantom, pos = 0, "\lrcorner"]
      & \ARR{\SET}
      \arrow[d, -{Triangle[open]}, "\Cod"]
      \\
      \SET
      \arrow[r, swap,"\Ar{\Parens{-}}"]
      & \SET
    \end{tikzcd}
  \]

\end{definition}

From a collection of atomic types $\Sorts$, we generate the type structure of the
$\lambda$-calculus as the least set $\Ty{\Sorts}$ closed under the following
formation rules:
\begin{mathpar}
  \inferrule{
    \tau\in\Sorts
  }{
    \tau\in\Ty{\Sorts}
  }
  \and
  \inferrule{
    \sigma\in\Ty{\Sorts}
    \\
    \tau\in\Ty{\Sorts}
  }{
    {\sigma\times\tau}\in\Ty{\Sorts}
  }
  \and
  \inferrule{
    \sigma\in\Ty{\Sorts}
    \\
    \tau\in\Ty{\Sorts}
  }{
    {\sigma\to\tau}\in\Ty{\Sorts}
  }
\end{mathpar}

\subsection{The clone of a \texorpdfstring{$\lambda$}{lambda}-signature}

From a $\lambda$-signature $\Sigma$ we can freely generate a special family of
sets $\Clone{\Sigma}{\Gamma,\tau}$ called its \emph{clone}, indexed in
$\Parens{\Gamma,\tau}\in\Ar{\Ty{\Sorts[\Sigma]}}$; simultaneously, we define the indexed set
of substitutions $\SbClone{\Sigma}{\Gamma,\Delta}$, indexed in
$\Parens{\Gamma,\Delta}\in\Ty{\Sorts[\Sigma]}^\star\times\Ty{\Sorts[\Sigma]}^\star$. In our
presentation, we choose to use \emph{explicit substitutions} rather than
implicit substitutions, because they are more abstract and allow a formulation
without explicit reference to preterms. Additionally, explicit substitutions
scale up to the metatheory of dependent type theory in a way that the implicit
(admissible) notion of substitution on preterms cannot.

We will write $\Seq{\Sigma}{\Gamma}{t}{\tau}$ and
$\Seq{\Sigma}{\Gamma}{\delta}{\Delta}$ to mean that
$t\in\Clone{\Sigma}{\Gamma,\tau}$ and
$\delta\in\SbClone{\Sigma}{\Gamma,\Delta}$ respectively, and
$\SeqEq{\Sigma}{\Gamma}{s}{t}{\tau}$ and
$\SeqEq{\Sigma}{\Gamma}{\delta_0}{\delta_1}{\Delta}$ to mean that $s$ and $t$
are equal elements of $\Clone{\Sigma}{\Gamma,\tau}$ and $\delta_0$ and
$\delta_1$ are equal elements of $\SbClone{\Sigma}{\Gamma,\Delta}$
respectively. In our notation, we \emph{presuppose} that $s$ and $t$ are
elements of the clone when we state that they are equal.

The clone of $\Sigma$ is defined as a quotient using the following indexed
inductive definition:

\begin{mathpar}
  \inferrule[variable]{}{
    \Seq{\Sigma}{\Gamma,\tau}{\Var}{\tau}
  }
  \and
  \inferrule[operation]{
    \Ops[\Sigma](\vartheta) \equiv (\Gamma,\tau)
  }{
    \Seq{\Sigma}{\Gamma}{\vartheta}{\tau}
  }
  \and
  \inferrule[subst]{
    \Seq{\Sigma}{\Gamma}{\delta}{\Delta}
    \\
    \Seq{\Sigma}{\Delta}{t}{\tau}
  }{
    \Seq{\Sigma}{\Gamma}{t[\delta]}{\tau}
  }
  \\
  \inferrule[abstraction]{
    \Seq{\Sigma}{\Gamma,\sigma}{t}{\tau}
  }{
    \Seq{\Sigma}{\Gamma}{
      \ObjLam{\sigma}{t}
    }{\sigma\to\tau}
  }
  \and
  \inferrule[application]{
    \Seq{\Sigma}{\Gamma}{s}{\sigma\to\tau}
    \\
    \Seq{\Sigma}{\Gamma}{t}{\sigma}
  }{
    \Seq{\Sigma}{\Gamma}{s(t)}{\tau}
  }
  \and
  \inferrule[pair]{
    \Seq{\Sigma}{\Gamma}{s}{\sigma}
    \\
    \Seq{\Sigma}{\Gamma}{t}{\tau}
  }{
    \Seq{\Sigma}{\Gamma}{(s,t)}{\sigma\times\tau}
  }
  \and
  \inferrule[proj1]{
    \Seq{\Sigma}{\Gamma}{s}{\sigma\times\tau}
  }{
    \Seq{\Sigma}{\Gamma}{s.1}{\sigma}
  }
  \and
  \inferrule[proj2]{
    \Seq{\Sigma}{\Gamma}{s}{\sigma\times\tau}
  }{
    \Seq{\Sigma}{\Gamma}{s.2}{\tau}
  }
  \\
  \inferrule[sb/idn]{}{
    \Seq{\Sigma}{\Gamma}{\SbIdn{\Gamma}}{\Gamma}
  }
  \and
  \inferrule[sb/proj]{}{
    \Seq{\Sigma}{\Gamma,\tau}{\SbProj}{\Gamma}
  }
  \and
  \inferrule[sb/ext]{
    \Seq{\Sigma}{\Gamma}{\delta}{\Delta}
    \\
    \Seq{\Sigma}{\Gamma}{t}{\tau}
  }{
    \Seq{\Sigma}{\Gamma}{\SbExt{\delta}{t}}{\Delta,\tau}
  }
  \and
  \inferrule[sb/cmp]{
    \Seq{\Sigma}{\Gamma}{\delta}{\Delta}
    \\
    \Seq{\Sigma}{\Delta}{\xi}{\Xi}
  }{
    \Seq{\Sigma}{\Gamma}{\xi\circ\delta}{\Xi}
  }
\end{mathpar}

Before we define the equivalence relation by which the clone is quotiented, it
will be useful to define an auxiliary substitution for De Bruijn lifting:
\begin{align*}
  \FmtTm{\SbLift{\gamma}} &\eqdef \FmtTm{\SbExt{(\SbProj\circ\gamma)}{\Var}}
\end{align*}

Next, we generate equivalence relations on terms and substitutions from the following rules:
\begin{mathpar}
  \inferrule[app/beta]{}{
    \SeqEq{\Sigma}{\Gamma}{
      \Parens{\ObjLam{\sigma}{t}}\Parens{s}
    }{
      t[\SbExt{\SbIdn{\Gamma}}{s}]
    }{\tau}
  }
  \and
  \inferrule[abs/eta]{}{
    \SeqEq{\Sigma}{\Gamma}{t}{
      \ObjLam{\sigma}{
        (t[\SbProj{}])(\Var)
      }
    }{\sigma\to\tau}
  }
  \and
  \inferrule[fst/beta]{}{
    \SeqEq{\Sigma}{\Gamma}{(s,t).1}{s}{\sigma}
  }
  \and
  \inferrule[snd/beta]{}{
    \SeqEq{\Sigma}{\Gamma}{(s,t).2}{t}{\tau}
  }
  \and
  \inferrule[pair/eta]{}{
    \SeqEq{\Sigma}{\Gamma}{t}{(t.1,t.2)}{\sigma\times\tau}
  }
  \and
  \inferrule[sb/cmp/idn/l]{}{
    \SeqEq{\Sigma}{\Gamma}{\SbIdn{\Delta}\circ\delta}{\delta}{\Delta}
  }
  \and
  \inferrule[sb/cmp/idn/r]{}{
    \SeqEq{\Sigma}{\Gamma}{\delta\circ\SbIdn{\Gamma}}{\delta}{\Delta}
  }
  \and
  \inferrule[sb/cmp/assoc]{}{
    \SeqEq{\Sigma}{\Gamma}{\gamma\circ(\delta\circ\xi)}{(\gamma\circ\delta)\circ\xi}{\Delta}
  }
  \and
  \inferrule[sb/cmp/proj]{}{
    \SeqEq{\Sigma}{\Gamma}{\SbProj\circ\SbExt{\delta}{t}}{\delta}{\Delta}
  }
  \and
  \inferrule[sb/cmp/dot]{}{
    \SeqEq{\Sigma}{\Gamma}{\delta\circ(\SbExt{\xi}{t})}{\SbExt{(\delta\circ\xi)}{t[\delta]}}{\Delta}
  }
  \and
  \inferrule[sb/var/idn]{}{
    \SeqEq{\Sigma}{\Gamma}{\Var[\SbIdn{\Gamma}]}{\Var}{\tau}
  }
  \and
  \inferrule[sb/var/ext]{}{
    \SeqEq{\Sigma}{\Gamma}{\Var[\SbExt{\delta}{t}]}{t}{\tau}
  }
  \and
  \inferrule[sb/abs]{}{
    \SeqEq{\Sigma}{\Gamma}{
      \Parens{\ObjLam{\sigma}{t}}[\delta]
    }{
      \ObjLam{\sigma}{
        t[\SbLift{\delta}]
      }
    }{\sigma\to\tau}
  }
  \and
  \inferrule[sb/app]{}{
    \SeqEq{\Sigma}{\Gamma}{(t(s))[\delta]}{t[\delta](s[\delta])}{\tau}
  }
  \and
  \inferrule[sb/pair]{}{
    \SeqEq{\Sigma}{\Gamma}{(s,t)[\delta]}{(s[\delta],t[\delta])}{\sigma\times\tau}
  }
  \and
  \inferrule[sb/proj1]{}{
    \SeqEq{\Sigma}{\Gamma}{t.1[\delta]}{(t[\delta]).1}{\sigma}
  }
  \and
  \inferrule[sb/proj2]{}{
    \SeqEq{\Sigma}{\Gamma}{t.2[\delta]}{(t[\delta]).2}{\tau}
  }
\end{mathpar}

We omit the congruence cases for brevity. The clone of $\Sigma$ is now defined
as the indexed family of quotients generated by the formation and definitional
equivalence rules given above.

\paragraph{Representation of the quotient}

In these notes, we will not dwell on the technical representation of the
quotiented terms. However, we will remark that the most convenient induction
principle for the quotiented syntax would arrive from a presentation as a
\emph{quotient inductive type}; moreover, our inductive definition falls under
a schema for \emph{finitary} quotient inductive types which is already known to
be interpretable in the setoid model of type
theory~\citep{dybjer-moeneclaey:2018}.

\subsection{The classifying category of a \texorpdfstring{$\lambda$}{lambda}-signature}

We can see that the language of substitutions above has the structure of
category; this category is in fact called $\ClCat{\Sigma}$, the
\emph{classifying category} or \emph{Lawvere category} of the
$\lambda$-signature $\Sigma$. The classifying category is also just
called the (pure) $\lambda$-\emph{theory} generated by the signature.
Concretely, it has contexts $\Gamma$ as objects, and equivalence classes of
substitutions $\Seq{\Sigma}{\Gamma}{\delta}{\Delta}$ as morphisms.

\begin{proposition}
  The classifying category $\ClCat{\Sigma}$ is cartesian closed.
\end{proposition}

\subsection{The category of renamings}\label{sec:ren-cat}

Every $\lambda$-signature gives rises to another category, namely the
\emph{category of renamings} $\Ren{\Sigma}$. Abstractly, this can be characterized as the
free strictly associative cartesian category generated by $\Ty{\Sorts[\Sigma]}$;
concretely, its objects are contexts $\Gamma$, and its morphisms
$\psi:\Hom[\Ren{\Sigma}]{\Gamma}{\Delta}$ are vectors of projections (indices) from $\Gamma$
into the types in $\Delta$.

An explicit presentation of $\Ren{\Sigma}$ appears
in~\cite{fiore:2002,fiore:2005} as the opposite of the comma category
$\lfloor-\rfloor\downarrow\mathsf{Const}(\Ty{\Sorts[\Sigma]})$, where
$\lfloor-\rfloor:\mathbb{F}\to\SET$ takes a finite cardinal to a set. Fiore
writes $\mathbb{F}\downarrow\Ty{\Sorts[\Sigma]}$ for this comma construction, and
$\mathbb{F}[\Ty{\Sorts[\Sigma]}]$ for its opposite. Another possible presentation of
$\Ren{\sigma}$ is as the subcategory of $\ClCat{\Sigma}$ which has the same
objects, but whose morphisms are all of the form $\SbIdn{\Gamma}$ or
$\SbExt{\SbExt{\SbIdn{\Gamma}}{t_0}}{\cdots{}t_n}$ for some $n> 0$, with $t_i$ of the
form $\Var[\SbProj^k]$, writing $\SbProj^k$ for the $k$-fold composition of
$\SbProj$ with itself.

\section{Normalization and the Yoneda embedding}\label{sec:yoneda}

Working in a constructive metatheory, we can see that the intensional content
of a certain natural isomorphism hides within it a normalization function for
the lambda calculus over $\Sigma$, as observed
in~\cite{cubric-dybjer-scott:1998}. Let $\Psh{\ClCat{\Sigma}}$ denote the
category of presheaves over the classifying category of $\Sigma$.

The Yoneda embedding is a cartesian closed functor
$\Yo:\ClCat{\Sigma}\to\Psh{\ClCat{\Sigma}}$, defined as $\Yo{\Delta} =
\Hom[\ClCat{\Sigma}]{-}{\Delta}$. Within the presheaf topos, it is easiest to think of
the representable objects $\Yo{\Delta}$ as the ``type of substitutions into
$\Delta$''.

There is another way to define the Yoneda embedding, which we will see is
naturally isomorphic to what is written above. In this version, we define a
functor $\FancyYo:\ClCat{\Sigma}\to\Psh{\ClCat{\Sigma}}$ by recursion on the
objects of $\ClCat{\Sigma}$. For an atomic type $\tau\in{}\Sorts[\Sigma]$, $\FancyYo\tau
= \Yo\tau$; but the remainder of the cases are defined using the cartesian
closed structure of the presheaf topos instead of the cartesian closed
structure of the classifying category:
\begin{align*}
  \FancyYo(\sigma\times\tau) &= \FancyYo{\sigma}\times\FancyYo{\tau}
  \\
  \FancyYo(\sigma\to\tau) &= {\FancyYo{\tau}}^{\FancyYo{\sigma}}
\end{align*}

Now, because the Yoneda embedding is cartesian closed, it is easy to see that
we have a natural isomorphism $\Yo\cong\FancyYo$. However, observe that the
elements in the fibers of $\FancyYo$ are not $\lambda$-terms, but a kind of
eta-long Boehm-tree representation of $\lambda$-terms.

That is, whereas the action of $\Yo$ on a syntactic morphism/term
$\Seq{\Sigma}{\Gamma}{t}{\sigma\times\tau}$ is to simply embed $t$ into the
appropriate presheaf fiber, the action of $\FancyYo$ on the same term must take
$t$ to an element of $\FancyYo{\sigma}\times\FancyYo{\tau}$, that is, an actual
\emph{pair}. Considering the case where $t$ is actually a variable, we can see
that the action of these two embeddings is intensionally quite different.
The other side of the natural isomorphism is witnessed by a ``readback''
operation, which reads one of these expanded Boehm trees into a syntactic term
(which can be seen to be $\beta$-normal and $\eta$-long). The normalization
operation obtained by composing these operations can be seen to be an instance
of \emph{normalization by evaluation}.

The problem with this kind of result, however, is that the categories have
quotiented too much for us to be able to say in mathematical (rather than
merely intuitive) language that we have obtained a normalization function. In
particular, the normalization operation that we describe above is actually
equal as a function to the identity. This is because the classifying category
is already quotiented by definitional equivalence.

As summarized in~\citet{streicher:1998}, there are two ways out of this
situation. One is to use a higher-dimensional structure, such as partial
equivalence relations or setoids, in order to structure the ambient category
theory; then, in the spirit of Bishop's constructive mathematics, we can
observe the intension of the normalization operation at the same time as seeing
that it is the identity in its extension. This approach was carried out
in~\cite{cubric-dybjer-scott:1998} using P-category theory, a variant of
E-category theory in which setoids are replaced by PERs.

Another more direct way is obtained from the \emph{gluing construction} in
category theory, where we will choose a different semantic domain which allows
us to see the difference between the two ways of interpreting syntax into the
presheaf category. This was carried out in detail
in~\citet{altenkirch-hofmann-streicher:1995}, but in a manner that is
unfortunately different enough from the classical construction that it is
unclear how it relates. In~\citet{fiore:2002}, normalization by evaluation for
typed lambda calculus is related explicitly to gluing; what we present in these
notes can be seen as an explicit instantiation of Fiore's frameork.

\section{Normalization by gluing}\label{sec:nbg}

To resolve the problem described above in Section~\ref{sec:yoneda}, we will
work with a more refined base category, namely the category of renamings
$\Ren{\Sigma}$ defined in Section~\ref{sec:ren-cat}. First observe that there
is an inclusion of categories $i : \Ren{\Sigma}\to\ClCat{\Sigma}$, since every
context renaming can be represented as a substitution, a sequence of extensions
by variables.

We have a reindexing functor $i^*:\Psh{\ClCat{\Sigma}}\to\Psh{\Ren{\Sigma}}$ by
precomposition. Composing with the Yoneda embedding, we can define a new
functor $\TM:\ClCat{\Sigma}\to\Psh{\Ren{\Sigma}}$:
\[
  \begin{tikzcd}[sep=huge,cramped]
    \ClCat{\Sigma}
    \arrow[d,hook,swap,"\Yo"]
    \arrow[dr,dashed,"\TM"]
    \\
    \Psh{\ClCat{\Sigma}}
    \arrow[r,swap,"i^*"]
    &
    \Psh{\Ren{\Sigma}}
  \end{tikzcd}
\]

\paragraph{Relative hom functor}

As described in~\citet{fiore:2002}, the functor $\TM$ is called the ``relative
hom functor'' of $i$, taking $\Delta:\ClCat{\Sigma}$ to
$\Hom[\ClCat{\Sigma}]{i(-)}{\Delta}$. In~\citet{fiore:2002}, this functor is
written $\langle{i}\rangle$, whereas we write $\TM$ in order to suggest the
intuition that it defines a presheaf of open terms.

We have constructed $\TM$ from the perspective of ``adjusting'' the Yoneda
embedding from $\ClCat{\Sigma}$, but~\citet{fiore:2002} explains another
characterization of the same functor from the perspective of the Yoneda
embedding from $\Ren{\Sigma}$. In particular, $\TM$ is the left Kan extension
of $\Yo:\Ren{\Sigma}\to\Ren{\Sigma}$ along $i$:
\[
  \begin{tikzcd}[cramped]
    \Ren{\Sigma}
    \arrow[rr,hook,"\Yo"]
    \arrow[dr, swap,"i"]
    &
    \arrow[d,phantom,"\Downarrow"]
    &
    \Psh{\Ren{\Sigma}}
    \\
    &
    \ClCat{\Sigma}
    \arrow[ur,dashed,swap,"\TM"]
  \end{tikzcd}
\]

\subsection{Presheaves of neutrals and normals}

In $\Psh{\Ren{\Sigma}}$, we can construct presheaves of neutral terms and
normal terms for each type; note that such presheaves cannot be defined in
$\Psh{\ClCat{\Sigma}}$, because they crucially cannot be closed under arbitrary
substitutions (whereas they happen to be closed under renamings). The fibers of
these presheaves will have the property that the equality relation for their
elements is \emph{discrete}.

To be concrete, let us begin by defining some restricted typing judgments for
neutrals and normals.
\begin{mathparpagebreakable}
  \inferrule[variable]{}{
    \SeqNE{\Sigma}{\Gamma}{\Var[\SbProj^k]}{\Gamma_{\vert{\Gamma}\vert-k-1}}
  }
  \and
  \inferrule[operation]{
    \Ops[\Sigma](\vartheta) \equiv (\Delta,\tau)
    \\
    \SeqNF{\Sigma}{\Gamma}{\delta}{\Delta}
  }{
    \SeqNE{\Sigma}{\Gamma}{\vartheta[\delta]}{\tau}
  }
  \and
  \inferrule[app]{
    \SeqNE{\Sigma}{\Gamma}{t}{\sigma\to\tau}
    \\
    \SeqNF{\Sigma}{\Gamma}{s}{\sigma}
  }{
    \SeqNE{\Sigma}{\Gamma}{t(s)}{\tau}
  }
  \and
  \inferrule[proj1]{
    \SeqNE{\Sigma}{\Gamma}{t}{\sigma\times\tau}
  }{
    \SeqNE{\Sigma}{\Gamma}{t.1}{\sigma}
  }
  \and
  \inferrule[proj2]{
    \SeqNE{\Sigma}{\Gamma}{t}{\sigma\times\tau}
  }{
    \SeqNE{\Sigma}{\Gamma}{t.2}{\tau}
  }
  \\
  \inferrule[shift]{
    \SeqNE{\Sigma}{\Gamma}{t}{\tau}
    \\
    \tau\in{}\Sorts[\Sigma]
  }{
    \SeqNF{\Sigma}{\Gamma}{t}{\tau}
  }
  \and
  \inferrule[abstraction]{
    \SeqNF{\Sigma}{\Gamma,\sigma}{t}{\tau}
  }{
    \SeqNF{\Sigma}{\Gamma}{
      \ObjLam{\sigma}{t}
    }{\sigma\to\tau}
  }
  \and
  \inferrule[pair]{
    \SeqNF{\Sigma}{\Gamma}{s}{\sigma}
    \\
    \SeqNF{\Sigma}{\Gamma}{t}{\tau}
  }{
    \SeqNF{\Sigma}{\Gamma}{(s,t)}{\sigma\times\tau}
  }
  \\
  \inferrule[sb/nf/proj]{}{
    \SeqNF{\Sigma}{\Gamma}{\SbProj^{\vert\Gamma\vert}}{[\,]}
  }
  \and
  \inferrule[sb/nf/ext]{
    \SeqNF{\Sigma}{\Gamma}{\delta}{\Delta}
    \\
    \SeqNF{\Sigma}{\Gamma}{t}{\tau}
  }{
    \SeqNF{\Sigma}{\Gamma}{\SbExt{\delta}{t}}{\Delta,\tau}
  }
  \\
  \inferrule[sb/ne/proj]{}{
    \SeqNE{\Sigma}{\Gamma}{\SbProj^{\vert\Gamma\vert}}{[\,]}
  }
  \and
  \inferrule[sb/ne/ext]{
    \SeqNE{\Sigma}{\Gamma}{\delta}{\Delta}
    \\
    \SeqNE{\Sigma}{\Gamma}{t}{\tau}
  }{
    \SeqNE{\Sigma}{\Gamma}{\SbExt{\delta}{t}}{\Delta,\tau}
  }
\end{mathparpagebreakable}

\paragraph{Admissible substitutions}

We have restricted the language of normal substitutions to consist in vectors
of terms, constructed using the \textsc{sb/proj} and \textsc{sb/ext} rules. The
identity substitution is \emph{admissible} as a neutral substitution, but is
not one of the generators. We define
$\SeqNE{\Sigma}{\Gamma}{\SbIdn{\Gamma}}{\Gamma}$ by recursion on $\Gamma$ as follows:
\begin{align*}
  \FmtNE{\SbIdn{[\,]}} &= \FmtNE{\SbProj^0}
  \\
  \FmtNE{\SbIdn{\Gamma,\tau}} &= \FmtNE{\SbExt{\SbIdn{\Gamma}}{\Var[\SbProj^0]}}
\end{align*}

\paragraph{\texorpdfstring{$\eta$}{Eta}-long normal forms}

Observe that we have ensured an $\eta$-long normal form by restricting the
\textsc{shift} rule above to apply only at atomic types. It is easy to see that
these judgments are closed under context renamings, i.e.\ support a
$\Ren{\Sigma}$-action.
Therefore, we can use these judgments as the raw material from which to build
the presheaves of neutrals and normals for each type $\tau\in\Ty{\Sorts[\Sigma]}$ as
follows:
\begin{align*}
  \NE{\tau} &: \Psh{\Ren{\Sigma}}
  \\
  \NE{\tau}(\Gamma) &\equiv \SetCompr{\FmtNE{t}}{\SeqNE{\Sigma}{\Gamma}{t}{\tau}}
  \\[6pt]
  \NF{\tau} &: \Psh{\Ren{\Sigma}}
  \\
  \NF{\tau}(\Gamma) &\equiv \SetCompr{\FmtNF{t}}{\SeqNF{\Sigma}{\Gamma}{t}{\tau}}
\end{align*}

\subsection{Syntax with binding, internally}

So far we have developed three presheaves of syntax in $\Psh{\Ren{\Sigma}}$:
the presheaf of typed terms $\TM(\tau)$, the presheaf of neutrals $\NE{\tau}$
and the presheaf of normals $\NF{\tau}$. Using the internal language of the
functor category, we can justify a simpler ``higher-order'' notation for
working with elements of these presheaves
internally~\citep{hofmann:1999,fiore-plotkin-turi:1999,staton:2007,harper-honsell-plotkin:1993}.

First observe that exponentiation of a presheaf
$\mathcal{F}:\Psh{\Ren{\Sigma}}$ by a representable has a simpler
characterization using the Yoneda lemma (in fact, this works for any base
category that has finite products):
\begin{align*}
  \mathcal{F}^{\Yo{\Delta}}(\Gamma)
  &\cong \Hom[\Psh{\Ren{\Sigma}}]{\Yo\Gamma\times\Yo\Delta}{\mathcal{F}}
  \\
  &\cong \Hom[\Psh{\Ren{\Sigma}}]{\Yo(\Gamma\times\Delta)}{\mathcal{F}}
  \\
  &\cong \mathcal{F}(\Gamma\times\Delta)
\end{align*}

Writing $\VAR{\tau}$ for the representable presheaf
$\Yo{\tau}:\Psh{\Ren{\Sigma}}$ of variables, we can equivalently use
a higher-order notation for terms from inside the topos, with constructors like
the following:
\begin{align*}
  \FmtNE{\Var} &: \VAR{\tau}\to\NE{\tau}
  \\
  \FmtNF{\lambda^{\Mute\sigma}} &: (\VAR{\sigma}\to\NF{\tau})\to\NF{\sigma\to\tau}
  \\
  &\ldots
\end{align*}

This is justified by the fact that all the generators of $\TM$, $\NE{}$ and
$\NF{}$ commute with the presheaf renaming action.  When working internally, we
will implicitly use these notations as a simplifying measure.

We will also employ the internal substitution constructors $\FmtNF{\SbEmp} :
\mathbf{1}\to\NF{[\,]}$ and $\FmtNE{\SbEmp} : \mathbf{1}\to\NE{[\,]}$ defined as
follows:
\begin{align*}
  \FmtNF{\SbEmp[\Gamma]}(\star) &= \FmtNF{\SbProj^{\vert\Gamma\vert}}
  \\
  \FmtNE{\SbEmp[\Gamma]}(\star) &= \FmtNE{\SbProj^{\vert\Gamma\vert}}
\end{align*}

\subsection{The gluing construction}

Next, we will construct the \emph{gluing category} which will serve as our
principal semantic domain for the model construction. Consider the comma
category $\GlCat{\Sigma} \equiv \Psh{\Ren{\Sigma}}\downarrow\TM$, which
``glues'' syntactic contexts together with their semantics in
presheaves.\footnote{Careful readers will note that this is a notation for the
actual instance of the comma construction,
$\mathbf{id}_{\Psh{\Ren{\Sigma}}}\downarrow\TM$.} Concretely, an object of
$\GlCat{\Sigma}$ is a tuple $(\mathcal{D}:\Psh{\Ren{\Sigma}},
\Delta:\ClCat{\Sigma}, \Quo{\Delta} : \mathcal{D}\to\TM(\Delta))$;
a morphism
$(\mathcal{G},\Gamma,\Quo{\Gamma})\to(\mathcal{D},\Delta,\Quo{\Delta})$ is a
commuting square of the following form, which we will suggestively write
$\Realize{\delta}{d}$:
\[
  \begin{tikzcd}[sep=huge,cramped]
    \mathcal{G}
    \arrow[r,"\FmtVal{d}"]
    \arrow[d,swap,"\Quo{\Gamma}"]
    &
    \mathcal{D}
    \arrow[d,"\Quo{\Delta}"]
    \\
    \TM(\Gamma)
    \arrow[r, swap, "\FmtTm{\delta}"]
    &
    \TM(\Delta)
  \end{tikzcd}
\]

The gluing category $\GlCat{\Sigma}$ is the category of proof-relevant logical
predicates, and is known to be cartesian closed, and thence a model of simply
typed lambda calculus. To use this information to our advantage, we will need
to ``unearth'' its cartesian closed structure in explicit terms.

\paragraph{Presentation as a pullback}

Following~\citet{frey:2013}, we can give a more intuitive presentation of the
gluing construction as a pullback of the fundamental fibration along $\TM$:
\[
  \begin{tikzcd}[sep=huge,cramped]
    {\GlCat{\Sigma}}
    \arrow[r, dashed]
    \arrow[d, -{Triangle[open]}, dashed,swap,"\GlFib{\Sigma}"]
    \arrow[dr, phantom, pos = 0, "\lrcorner"]
    &
    \ARR{\Psh{\Ren{\Sigma}}}
    \arrow[d,-{Triangle[open]}, "\Cod"]
    \\
    \ClCat{\Sigma}
    \arrow[r,swap,"\TM"]
    &
    \Psh{\Ren{\Sigma}}
  \end{tikzcd}
\]

From the pullback above, we have the \emph{gluing fibration}
$\GlFib{\Sigma}:\GlCat{\Sigma}\to\ClCat{\sigma}$ which acts on objects
$(\mathcal{D},\Delta,\Quo{\Delta})$ by projecting $\Delta$, and on morphisms
$\Realize{\delta}{d}:\Hom[\GlCat{\Sigma}]{(\mathcal{G},\Gamma,q_\Gamma)}{(\mathcal{D},\Delta,\Quo{\Delta})}$ by projecting
$\delta$.\footnote{\citet{streicher:1998} calls this the ``codomain functor'',
but to avoid confusion with the codomain functor that it is a pullback of, we
use a different terminology.}

\subsection{Reification, reflection and logical predicates}

Observe that there are obvious natural embeddings
$\RbNF{\tau}:\NF{\tau}\hookrightarrow\TM(\tau)$ and
$\RbNE{\tau}:\NE{\tau}\hookrightarrow\TM(\tau)$ for each
$\tau\in\Ty{\Sorts[\Sigma]}$, called ``readback''.

In order to give an explicit character to the cartesian closed structure of
$\GlCat{\Sigma}$, we will define a proof-relevant family of logical predicates
$\Pred{\tau}:\Psh{\Ren{\Sigma}}$ by induction on $\tau\in\Ty{\Sorts[\Sigma]}$,
simultaneously exhibiting natural transformations
$\Reflect{\tau}:\NE{\tau}\to\Pred{\tau}$ (pronounced ``reflect'') and
$\Reify{\tau}:\Pred{\tau}\to\NF{\tau}$ (pronounced ``reify'') such that the
following triangle commutes:
\[
  \begin{tikzcd}[sep=huge,cramped]
    \NE{\tau}
    \arrow[rr,"\Reify{\tau}\circ\Reflect{\tau}"]
    \arrow[dr,hook,swap,"\RbNE{\tau}"]
    &&
    \NF{\tau}
    \arrow[dl,hook,"\RbNF{\tau}"]
    \\
    &\TM(\tau)
  \end{tikzcd}
  \tag{reify-reflect yoga}
\]

\begin{remark}

  An alternative to this approach is to
  follow~\citet{altenkirch-hofmann-streicher:1995} and employ an ad hoc
  ``twisted gluing'' category, in which the data of the gluing objects contains
  the reification and reflection maps. This has the benefit of leading to a
  proof which is more self-contained, but the disadvantage is that it is not
  clear how to connect this twisted gluing category to the classical
  construction.

\end{remark}

\paragraph{Atomic types}
For an atomic type $\sigma\in{}\Sorts[\Sigma]$, we define $\Pred{\sigma}=\NF{\sigma}$,
$\Reflect{\tau}=1$, $\Reify{\tau}=1$; it is easy to see that the
reify-reflect yoga is upheld. Next, we come to the compound types.

\paragraph{Product types}
Fixing types $\sigma,\tau\in\Ty{\Sorts[\Sigma]}$, we define the logical predicate and the reflection and reification maps, using the internal language of $\Psh{\Ren{\Sigma}}$:
\begin{align*}
  \Pred{\sigma\times\tau} &= \Pred{\sigma}\times\Pred{\tau}
  \\
  \Reflect{\sigma\times\tau}[t]
  &=
  \FmtVal{
    \Parens{
      \Reflect{\sigma}[t.1], \Reflect{\tau}[t.2]
    }
  }
  \\
  \Reify{\sigma\times\tau}[\FmtVal{v_0},\FmtVal{v_1}]
  &=
  \FmtNF{
    \Parens{
      \Reify{\sigma}[v_0],
      \Reify{\tau}[v_1]
    }
  }
\end{align*}

To execute the reify-reflect yoga, working internally, we fix
$t:\NE{\sigma\times\tau}$; we need to observe that
$\RbNF{\sigma\times\tau}[\Reflect{\sigma\times\tau}[\Reify{\sigma\times\tau}[t]]]
= \RbNE{\sigma\times\tau}[t]$.
\begin{align*}
  \RbNF{\sigma\times\tau}[\Reify{\sigma\times\tau}[\Reflect{\sigma\times\tau}[t]]]
  &=
  \RbNF{\sigma\times\tau}[
    \Reify{\sigma\times\tau}[
      \Reflect{\sigma}[t.1], \Reflect{\tau}[t.2]
    ]
  ]
  \\
  &=
  \RbNF{\sigma\times\tau}[
    \Reify{\sigma}[\Reflect{\sigma}[t.1]],
    \Reify{\tau}[\Reflect{\tau}[t.2]]
  ]
  \\
  &=
  \FmtTm{
    \Parens{
      \RbNF{\sigma}[\Reify{\sigma}[\Reflect{\sigma}[t.1]]],
      \RbNF{\tau}[\Reify{\tau}[\Reflect{\tau}[t.1]]]
    }
  }
  \\
  &=
  \FmtTm{
    \Parens{
      \RbNE{\sigma}[t.1],
      \RbNE{\tau}[t.2]
    }
  }
  \tag{i.h., i.h.}
  \\
  &=
  \FmtTm{
    \Parens{
      (\RbNE{\sigma\times\tau}[t]).1,
      (\RbNE{\sigma\times\tau}[t]).2
    }
  }
  \\
  &=
  \RbNE{\sigma\times\tau}[t]
  \tag{\textsc{pair/eta}}
\end{align*}

Above, the steps that commute readback of (neutrals, normals) with the syntax
of the $\lambda$-theory follow from the fact that normals and neutrals actually
embed directly into the syntax unchanged.

\paragraph{Function types}

To interpret function types, we cannot simply use the exponential in
$\Psh{\Ren{\Sigma}}$, as this would take us outside the realm of
\emph{definable} functions. In a move apparently inspired by Kreisel's
\emph{modified realizability}, we include in the logical predicate both a
definable function and its meaning, taking the pullback
\[
  \begin{tikzcd}[cramped]
    {\Pred{\sigma\to\tau}}
    \arrow[r,dashed]
    \arrow[d,dashed]
    \arrow[dr, phantom, pos = 0, "\lrcorner"]
    &
    {\Pred{\tau}}^{\Pred{\sigma}}
    \arrow[d,"\FmtVal\phi"]
    \\
    \TM(\sigma\to\tau)
    \arrow[r,swap,"\FmtVal\psi"]
    & \Parens{\TM(\tau)}^{\Pred{\sigma}}
  \end{tikzcd}
\]
where for clarity, we define arrows $\FmtVal{\phi},\FmtVal{\psi}$ in the internal language of $\Psh{\Ren{\Sigma}}$ as follows:
\begin{align*}
  \FmtVal{F}:{\Pred{\tau}}^{\Pred{\sigma}}\vdash\FmtVal{\phi} & \equiv
  \MetaLam{\FmtVal{v}}{
    \RbNF{\tau}[\Reify{\tau}[F(v)]]
  }
  \\
  \FmtTm{t}:\TM(\sigma\to\tau)\vdash\FmtVal{\psi}&\equiv
  \MetaLam{\FmtVal{v}}{
    \FmtTm{
      t\Parens{
        \RbNF{\sigma}[\Reify{\sigma}[v]]
      }
    }
  }
\end{align*}

Abusing notation slightly, we will write an element of $\Pred{\sigma\to\tau}$
as $\Realize{t}{F}$ where $\FmtTm{t}:\TM(\sigma\to\tau)$ and
$\FmtVal{F}:{\Pred{\tau}}^{\Pred{\sigma}}$.
Next, we need to define reflection of neutrals and reification into normals:
\begin{align*}
  \Reflect{\sigma\to\tau}[t]
  &=
  \Realize{
    \RbNE{\sigma\to\tau}[t]
  }{
    \MetaLam{\FmtVal{v}}{
      \Reflect{\tau}[t(\Reify{\sigma}[v])]
    }
  }
  \\
  \Reify{\sigma\to\tau}[\Realize{t}{F}]
  &=
  \FmtNF{
    \ObjLam{\sigma}{
      \MetaLam{\FmtNF{x}}{
        \Reify{\tau}[
          F\Parens{
            \Reflect{\sigma}[\Var(x)]
          }
        ]
      }
    }
  }
\end{align*}

To prove the reify-reflect yoga, (working internally) fix
$t:\NE{\sigma\to\tau}$.
\begin{align*}
  \RbNF{\sigma\to\tau}[\Reify{\sigma\to\tau}[\Reflect{\sigma\to\tau}[t]]]
  &=
  \RbNF{\sigma\to\tau}[
    \Reify{\sigma\to\tau}[
      \Realize{
        \RbNE{\sigma\to\tau}[t]
      }{
        \MetaLam{v}{
          \Reflect{\tau}[
            t(\Reify{\sigma}[v])
          ]
        }
      }
    ]
  ]
  \\
  &=
  \RbNF{\sigma\to\tau}[
    \ObjLam{\sigma}{
      \MetaLam{x}{
        \Reify{\tau}[
          \Reflect{\tau}[
            t\Parens{
              \Reify{\sigma}[\Reflect{\sigma}[\Var(x)]]
            }
          ]
        ]
      }
    }
  ]
  \\
  &=
  \FmtTm{
    \ObjLam{\sigma}{
      \MetaLam{x}{
        \RbNF{\tau}[
          \Reify{\tau}[
            \Reflect{\tau}[
              t\Parens{
                \Reify{\sigma}[
                  \Reflect{\sigma}[\Var(x)]
                ]
              }
            ]
          ]
        ]
      }
    }
  }
  \\
  &=
  \FmtTm{
    \ObjLam{\sigma}{
      \MetaLam{x}{
        \RbNE{\tau}[
          t\Parens{
            \Reify{\sigma}[
              \Reflect{\sigma}[\Var(x)]
            ]
          }
        ]
      }
    }
  }
  \tag{i.h.}
  \\
  &=
  \FmtTm{
    \ObjLam{\sigma}{
      \MetaLam{x}{
        \Parens{\RbNE{\tau}[t]}\Parens{
          \RbNF{\sigma}[
            \Reify{\sigma}[
              \Reflect{\sigma}[\Var(x)]
            ]
          ]
        }
      }
    }
  }
  \\
  &=
  \FmtTm{
    \ObjLam{\sigma}{
      \MetaLam{x}{
        \Parens{\RbNE{\tau}[t]}\Parens{\RbNE{\sigma}[\Var(x)]}
      }
    }
  }
  \\
  &=
  \FmtTm{
    \ObjLam{\sigma}{
      \MetaLam{x}{
        \Parens{\RbNE{\tau}[t]}\Parens{\Var(x)}
      }
    }
  }
  \tag{i.h.}
  \\
  &=
  \RbNE{\tau}[t]
  \tag{\textsc{abs/eta}}
\end{align*}

\paragraph{Contexts}

The interpretation is now extended to contexts $\Gamma\in\Ty{\Sorts[\Sigma]}^\star$,
which are the ``types'' of substitutions; the interpretation is essentially the
same as the one for products.
\begin{align*}
  \Pred{[\,]} &= \mathbf{1}
  \\
  \Reflect{[\,]}[\SbEmp] &= \FmtVal{\star}
  \\
  \Reify{[\,]}[\star] &= \FmtNF{\SbEmp}
  \\[6pt]
  \Pred{\Gamma,\tau} &= \Pred{\Gamma}\times\Pred{\tau}
  \\
  \Reflect{\Gamma,\tau}[\SbExt{\gamma}{t}] &= (\Reflect{\Gamma}[\gamma], \Reflect{\tau}[t])
  \\
  \Reify{\Gamma,\tau}[g,v] &= \SbExt{\Reify{\Gamma}[g]}{\Reify{\tau}[v]}
\end{align*}

The reify-reflect yoga follows in exactly the same way as it did for products.

\bigskip
\bigskip

Observe that for any $\tau\in\Ty{\Sorts[\Sigma]}$, the triple
$\Sem{\tau}\equiv \Parens{\Pred{\tau},\tau,\Quo{\tau}\equiv\RbNF{\tau}\circ\Reify{\tau}}$ is an object in
$\GlCat{\Sigma}$.
This brings us to an explicit characterization of the cartesian closed
structure of $\GlCat{\Sigma}$.

\begin{theorem}
   For $\sigma,\tau\in\Ty{\Sorts[\Sigma]}$, $\Sem{\sigma\times{}\tau}$ is the
   cartesian product $\Sem{\sigma}\times\Sem{\tau}$ in $\GlCat{\Sigma}$.
\end{theorem}

\begin{proof}
  We will establish that $\Sem{\sigma\times\tau}$ is the cartesian product
  $\Sem{\sigma}\times\Sem{\tau}$ by exhibiting its universal property. We
  need to exhibit a span in $\GlCat{\Sigma}$ with the following property for any
  $D,d_1,d_2$:
  \[
    \begin{tikzcd}[sep=large,cramped]
      D
      \arrow[dr,dashed,"{\exists!\bar{d}}"]
      \arrow[ddr,swap,bend right=20,"d_1"]
      \arrow[drr,bend left=20,"d_2"]
      \\
      &
      \Sem{\sigma\times\tau}
      \arrow[d,"\pi_1"]
      \arrow[r,"\pi_2"]
      &
      \Sem{\tau}
      \\
      &
      \Sem{\sigma}
    \end{tikzcd}
  \]

  The projections $\pi_1,\pi_2$ are the following commuting squares:
  \begin{mathpar}
    \begin{tikzcd}[sep=large]
      \Pred{\sigma}\times\Pred{\tau}
      \arrow[r,"{\FmtVal{(p,q)}\vdash{}\FmtVal{p}}"]
      \arrow[d,swap,"\Quo{\sigma\times\tau}"]
      &
      \Pred{\sigma}
      \arrow[d,swap,"\Quo{\sigma}"]
      \\
      \TM(\sigma\times\tau)
      \arrow[r,swap,"{\FmtTm{t}\vdash\FmtTm{t.1}}"]
      &
      \TM(\sigma)
    \end{tikzcd}
    \and
    \begin{tikzcd}[sep=large]
      \Pred{\sigma}\times\Pred{\tau}
      \arrow[r,"{\FmtVal{(p,q)}\vdash\FmtVal{q}}"]
      \arrow[d,swap,"\Quo{\sigma\times\tau}"]
      &
      \Pred{\tau}
      \arrow[d,swap,"\Quo{\tau}"]
      \\
      \TM(\sigma\times\tau)
      \arrow[r,swap,"{\FmtTm{t}\vdash{}\FmtTm{t.2}}"]
      &
      \TM(\tau)
    \end{tikzcd}
  \end{mathpar}

  We show that the first square commutes (the second is identical); fixing
  $\FmtVal{p}:\Pred{\sigma},\FmtVal{q}:\Pred{\tau}$, we calculate.
  \begin{align*}
    \FmtTm{
      \Parens{
        \RbNF{\sigma\times\tau}[\Reify{\sigma\times\tau}[p,q]]
      }.1
    }
    &=
    \FmtTm{
      \Parens{
        \RbNF{\sigma\times\tau}[
          \Reify{\sigma}[p],
          \Reify{\tau}[q]
        ]
      }.1
    }
    \tag{def.}
    \\
    &=
    \FmtTm{
      \Parens{
        \RbNF{\sigma}[\Reify{\sigma}[p]],
        \RbNF{\tau}[\Reify{\tau}[q]]
      }.1
    }
    \tag{def.}
    \\
    &=
    \FmtTm{
      \Parens{
        \RbNF{\sigma}[\Reify{\sigma}[p]]
      }
    }
    \tag{\textsc{fst/beta}}
  \end{align*}

  Next, we need to show that there is a unique mediating arrow
  $\bar{d}:D\to\Sem{\sigma\times\tau}$ such that the two triangles commute.
  Unfolding what we are given, we have $D\equiv(\mathcal{D},\Delta,\Quo{\Delta})$
  and two commuting squares:
  \begin{mathpar}
    \begin{tikzcd}[sep=large,cramped]
      \mathcal{D}
      \arrow[r,"\FmtVal{d_{00}}"]
      \arrow[d,swap,"\Quo{\Delta}"]
      &
      \Pred{\sigma}
      \arrow[d,swap,"\Quo{\sigma}"]
      \\
      \TM(\Delta)
      \arrow[r,swap,"\FmtTm{d_{01}}"]
      &
      \TM(\sigma)
    \end{tikzcd}
    \and
    \begin{tikzcd}[sep=large,cramped]
      \mathcal{D}
      \arrow[r,"\FmtVal{d_{10}}"]
      \arrow[d,swap,"\Quo{\Delta}"]
      &
      \Pred{\tau}
      \arrow[d,swap,"\Quo{\tau}"]
      \\
      \TM(\Delta)
      \arrow[r,swap,"\FmtTm{d_{11}}"]
      &
      \TM(\tau)
    \end{tikzcd}
  \end{mathpar}

  We define the mediating map $\bar{d}$ as the following square:
  \[
    \begin{tikzcd}[sep=large,cramped]
      \mathcal{D}
      \arrow[r,"\FmtVal{(d_{00},d_{10})}"]
      \arrow[d,swap,"\Quo{\Delta}"]
      &
      \Pred{\sigma\times\tau}
      \arrow[d,swap,"\Quo{\sigma\times\tau}"]
      \\
      \TM(\Delta)
      \arrow[r,swap,"\FmtTm{(d_{01},d_{11})}"]
      &
      \TM(\sigma\times\tau)
    \end{tikzcd}
  \]

  To see that the square commutes, fix $\FmtVal{p}:\mathcal{D}$ and calculate.
  \begin{align*}
    \RbNF{\sigma\times\tau}[
      \Reify{\sigma\times\tau}[
        d_{00}(p),d_{10}(p)
      ]
    ]
    &=
    \RbNF{\sigma\times\tau}[
      \Reify{\sigma}[d_{00}(p)],
      \Reify{\tau}[d_{10}(p)]
    ]
    \\
    &=
    \FmtTm{
      \Parens{
        \RbNF{\sigma}[
          \Reify{\sigma}[d_{00}(p)]
        ],
        \RbNF{\sigma}[
          \Reify{\sigma}[d_{00}(p)]
        ]
      }
    }
    \\
    &=
    \FmtTm{
      \Parens{
        d_{01}(\Quo{\Delta}[p]),
        d_{11}(\Quo{\Delta}[p])
      }
    }
  \end{align*}

  It is easy to see that $\pi_1\circ\bar{d}=d_1$ and $\pi_2\circ\bar{d}=d_2$.
  The uniqueness of $\bar{d}$ with this property follows from the fact that its
  components are unique: $\FmtVal{(d_{00},d_{10})}$ is the unique mediating arrow given
  by the universal property of the product
  $\Pred{\sigma\times\tau}\equiv\Pred{\sigma}\times\Pred{\tau}$; moreover,
  because $\TM$ preserves finite products, we can say the same of
  $\FmtTm{(d_{01},d_{11})}$.\footnote{Recall that $\TM$ is defined as $i^*\circ\Yo$
  with $i^*$ the reindexing functor induced by
  $i:\Ren{\Sigma}\to\ClCat{\Sigma}$. Because $i^*$ has a left adjoint, given by
  Kan extension, it preserves limits; because the Yoneda embedding also
  preserves limits, $\TM$ preserves limits too.}

\end{proof}

\begin{exercise}
   For $\sigma,\tau\in\Ty{\Sorts[\Sigma]}$, show that $\Sem{\sigma\to\tau}$ is the
   exponential in $\Sem{\tau}^{\Sem{\sigma}}$ in
   $\GlCat{\Sigma}$~\citep[Example~2.1.12]{johnstone:2002}.
\end{exercise}

\begin{corollary}\label{cor:gluing-model}
  $\GlCat{\Sigma}$ is a model of the free $\lambda$-theory generated by
  $\Sigma$, with interpretation functor $\Sem{-}:\ClCat{\Sigma}\to\GlCat{\Sigma}$.
\end{corollary}

\begin{theorem}\label{thm:eval-cod}
  The composite functor $\GlFib{\Sigma}\circ\Sem{-}$ is the identity endofunctor on $\ClCat{\Sigma}$:
  \[
    \begin{tikzcd}[sep=huge]
      \ClCat{\Sigma}
      \arrow[r,"\Sem{-}"]
      \arrow[swap,dr,"\mathbf{id}_{\ClCat{\Sigma}}"]
      &
      \GlCat{\Sigma}
      \arrow[d,"\GlFib{\Sigma}"]
      \\
      &
      \ClCat{\Sigma}
    \end{tikzcd}
  \]
\end{theorem}
\begin{proof}

  This follows immediately from the fact that $\ClCat{\Sigma}$ is the
  classifying category of the theory $\Sigma$, so it is the initial category
  with the structure of $\Sigma$. Therefore, any $\Sigma$-homomorphism
  $\ClCat{\Sigma}\to\ClCat{\Sigma}$ must be the identity, including the composite above.
\end{proof}

Now, working externally in the category $\SET$, we can explicitly construct the
normalization function,
$\NfFun{\Gamma}{\Delta}:\Hom[\ClCat{\Sigma}]{\Gamma}{\Delta}\to\NF{\Delta}(\Gamma)$
as the following composite:
\[
  \begin{tikzcd}
    \Hom[\ClCat{\Sigma}]{\Gamma}{\Delta}
    \arrow[swap,r,"\Sem{-}"]
    \arrow[rrrr,bend left=20,"\NfFun{\Gamma}{\Delta}"]
    &
    \Hom[\GlCat{\Sigma}]{\Sem{\Gamma}}{\Sem{\Delta}}
    \arrow[swap,r, "\pi"]
    &
    \Hom[\Psh{\Ren{\Sigma}}]{\Pred{\Gamma}}{\Pred{\Delta}}
    \arrow[swap,r, "\phi"]
    &
    \Pred{\Delta}(\Gamma)
    \arrow[swap,r, "{\Reify[\Gamma]{\Delta}}"]
    &
    \NF{\Delta}(\Gamma)
  \end{tikzcd}
\]
where
\begin{align*}
  \phi(\FmtVal{v}) &= \FmtVal{v}_{\Mute\Gamma}(\Reflect[\Gamma]{\Gamma}[\SbIdn[\Gamma]{\Gamma}])
\end{align*}

\begin{theorem}[Completeness]
  If $\SeqEq{\Sigma}{\Gamma}{t_0}{t_1}{\tau}$, then
  $\SeqEq{\Sigma}{\Gamma}{\RbNF[\Gamma]{\tau}[\NfFun{\Gamma}{\tau}[t_0]]}{\RbNF[\Gamma]{\tau}[\NfFun{\Gamma}{\tau}[t_1]]}{\tau}$.
\end{theorem}
\begin{proof}
  This is immediate from the fact that we have defined a function out of the
  morphisms of $\ClCat{\Sigma}$, which are already quotiented under
  definitional equivalence.
\end{proof}

\begin{theorem}[Normalization]\label{thm:normalization}
  If $\Seq{\Sigma}{\Gamma}{t}{\tau}$, then $\SeqEq{\Sigma}{\Gamma}{\RbNF[\Gamma]{\tau}[\NfFun{\Gamma}{\tau}[t]]}{t}{\tau}$.
\end{theorem}
\begin{proof}
  Suppose $\FmtVal{\Sem{t}} = \FmtVal{\Parens{\Realize{t_0}{v}}}$. Now calculate.
  \begin{align*}
    \RbNF[\Gamma]{\tau}[\NfFun{\Gamma}{\tau}[t]]
    &=
    \RbNF[\Gamma]{\tau}[
      \Reify[\Gamma]{\tau}[
        v_{\Mute\Gamma}\Parens{
          \Reflect[\Gamma]{\Gamma}[\SbIdn[\Gamma]{\Gamma}]
        }
      ]
    ]
    \\
    &=
    \Parens{
      \RbNF{\tau}\circ\Reify{\tau}\circ{}\FmtVal{v}
    }_{\Mute\Gamma}(\Reflect[\Gamma]{\Gamma}[\SbIdn[\Gamma]{\Gamma}])
    \\
    &=
    \Parens{
      \FmtTm{t_0}\circ\RbNF{\Gamma}\circ\Reify{\Gamma}
    }_{\Mute\Gamma}(\Reflect[\Gamma]{\Gamma}[\SbIdn[\Gamma]{\Gamma}])
    \tag{comma condition}
    \\
    &=
    \FmtTm{
      t_0^{\Mute\Gamma}\Parens{
        \RbNF[\Gamma]{\Gamma}[
          \Reify[\Gamma]{\Gamma}(\Reflect[\Gamma]{\Gamma}(\SbIdn[\Gamma]{\Gamma}))
        ]
      }
    }
    \\
    &=
    \FmtTm{
      t_0^{\Mute\Gamma}\Parens{
        \RbNE[\Gamma]{\Gamma}[\SbIdn[\Gamma]{\Gamma}]
      }
    }
    \tag{reify-reflect yoga}
    \\
    &=
    \FmtTm{
      t_0^{\Mute\Gamma}\Parens{
        \SbIdn[\Gamma]{\Gamma}
      }
    }
  \end{align*}

  $\SeqEq{\Sigma}{\Gamma}{t_0^{\Mute\Gamma}(\SbIdn[\Gamma]{\Gamma})}{t}{\tau}$.  Writing
  $\FmtTm{\lfloor{}t\rfloor}$ for the induced natural transformation
  $\TM(\Gamma)\to\TM(\tau)$ such that
  $\SeqEq{\Sigma}{\Gamma}{\lfloor{}t\rfloor_{\Mute\Gamma}(\SbIdn[\Gamma]{\Gamma})}{t}{\tau}$, it
  suffices to show that $\FmtTm{t} = \FmtTm{t_0}$. Because $\GlFib{\Sigma}\Parens{\Sem{\FmtTm{t}}}=\FmtTm{t_0}$, by
  Theorem~\ref{thm:eval-cod} we have $\FmtTm{t}=\FmtTm{t_0}$.
\end{proof}

\begin{corollary}[Soundness]

  If
  $\SeqEq{\Sigma}{\Gamma}{\RbNF[\Gamma]{\tau}[\NfFun{\Gamma}{\tau}[t_0]]}{\RbNF[\Gamma]{\tau}[\NfFun{\Gamma}{\tau}[t_1]]}{\tau}$,
  then $\SeqEq{\Sigma}{\Gamma}{t_0}{t_1}{\tau}$.

\end{corollary}

\begin{proof}
  To see that $\FmtTm{t_0}=\FmtTm{t_1}$, observe that by Theorem~\ref{thm:normalization} we have $\RbNF[\Gamma]{\tau}[\NfFun{\Gamma}{\tau}[t_i]] = \FmtTm{t_i}$, so by transitivity and assumption we have $\FmtTm{t_0}=\FmtTm{t_1}$.
\end{proof}

\section{Perspective}

\subsection{Global sections and the Freyd cover}
\label{sec:freyd-cover}

A more common use of the gluing technique lies in the construction of the Freyd
cover (also called the ``scone'', which is short for ``Sierpinski cone'') of a
topos in order to prove properties of \emph{closed proofs} in intuitionistic
higher-order logic, such as the disjunction and existence properties, which
correspond in $\lambda$-calculus to instances of the closed canonicity
result~\citep[p.\ 228]{lambek-scott:1986}.

As an example, we will prove both these properties for intuitionistic
higher-order logic over simple types and the natural numbers.\footnote{This
section is an expanded version of material which appears
in~\citet{shulman:blog:scones-logical-relations}, with some more details filled in.}
Writing $\mathcal{F}$ for the free topos generated by a natural numbers object
$\mathsf{N}$, observe that we have the global sections functor
$\Gamma\equiv\mathcal{F}(\mathbf{1},-)$ which takes every object to its global
elements.
We define the Freyd cover $\breve{\mathcal{F}}$ over $\mathcal{F}$ as the gluing
category obtained by pulling back the fundamental fibration along the global
sections functor:
\[
  \begin{tikzcd}[sep=huge,cramped]
    {\breve{\mathcal{F}}}
    \arrow[r,dashed]
    \arrow[d,-{Triangle[open]}, dashed,swap,"\pi_1"]
    \arrow[dr, phantom, pos = 0, "\lrcorner"]
    &
    {\SET}^\to
    \arrow[d,-{Triangle[open]}, "\Cod"]
    \\
    \mathcal{F}
    \arrow[r,swap,"\Gamma"]
    &
    \SET
  \end{tikzcd}
\]

Because the global sections functor preserves finite limits, the Freyd cover
$\breve{\mathcal{F}}$ is again a topos with $\pi_1$ a \emph{logical
functor}~\citep[Example~2.1.12]{johnstone:2002}; moreover, $\pi_1$ preserves
the natural numbers object~\citep[Corollary~7.7.2]{taylor:1999}. We also have a
functor $\pi_0:\breve{\mathcal{F}}\to\SET$, which merely preserves finite
limits.

The Freyd cover $\breve{\mathcal{F}}$ has a natural numbers object
$\breve{\mathsf{N}}$ given by $(\mathsf{N},\mathbb{N},n\mapsto\bar{n})$, where
$\bar{n}$ takes a set-theoretic natural number to the corresponding global section in
$\mathcal{F}$, and $\pi_1$ preserves the natural numbers object.
Because $\mathcal{F}$ is the initial topos with a natural numbers object, for
any other such topos $\mathcal{E}$ we have a \emph{unique} map $I_{\mathcal{E}}
: \mathcal{F}\to\mathcal{E}$.

\begin{lemma}\label{lem:gluing-retract}
  The logical functor $\pi_1:\breve{\mathcal{F}}\to\mathcal{F}$ is a retract of $I_{\breve{\mathcal{F}}}$:
  \[
    \begin{tikzcd}[sep=large,cramped]
      \mathcal{F}
      \arrow[r,"I_{\breve{\mathcal{F}}}"]
      \arrow[dr, swap,"1"]
      &
      \breve{\mathcal{F}}
      \arrow[d,"\pi_1"]
      \\
      &
      \mathcal{F}
    \end{tikzcd}
  \]
\end{lemma}

\begin{proof}
  We have the ``additional'' identity morphism $1:\mathcal{F}\to\mathcal{F}$, so by initiality of $\mathcal{F}$, we must have $1=\pi_1\circ{}I_{\breve{\mathcal{F}}}$.
\end{proof}

\begin{theorem}[Natural number canonicity]
  Any global section $n:\Gamma(\mathsf{N})$ in $\mathcal{F}$ is equal to some
  numeral $\bar{k}$.
\end{theorem}

\begin{proof}
  The functor $I_{\breve{\mathcal{F}}}$ necessarily preserves $\mathsf{N}$.
  Therefore, the global section $n$ lifts in $\breve{\mathcal{F}}$ to a square
  in $\SET$ as follows:
  \[
    \begin{tikzcd}[sep=large,cramped]
      \mathbf{1}
      \arrow[r]
      \arrow[d]
      &
      \mathbb{N}
      \arrow[d]
      \\
      \Gamma(\mathbf{1})\cong\mathbf{1}
      \arrow[r]
      &
      \Gamma(\mathsf{N})
    \end{tikzcd}
  \]

  The upstairs morphism gives us a numeral $k$; because the diagram commutes
  and using Lemma~\ref{lem:gluing-retract}, we have $n=\bar{k}$.

\end{proof}

\begin{theorem}[Existence property]\label{thm:existence-property}

  Suppose that $\mathcal{F}\models\exists x:X.\phi(x)$; then there is a global
  element $\alpha:\mathbf{1}\to{}X$ in $\mathcal{F}$ such that
  $\mathcal{F}\models\phi(\alpha)$.

\end{theorem}

\begin{proof}

  We will use the Kripke-Joyal semantics of the
  topos~\citep{maclane-moerdijk:1992}; unwinding our assumption
  $\mathbf{1}\Vdash\exists x:X.\phi(x)$, we have that there exists an
  epimorphism $p:V\twoheadrightarrow\mathbf{1}$ and a morphism $\beta:V\to{}X$
  such that $V\Vdash\phi(\beta)$.\footnote{Please note that the symbol $\Vdash$
  here denotes the forcing relation, rather than the gluing relation.}

  The logical functor
  $I_{\breve{\mathcal{F}}}:\mathcal{F}\to\breve{\mathcal{F}}$ lifts $p$ to an epimorphism $I_{\breve{\mathcal{F}}}(p) :
  I_{\breve{\mathcal{F}}}(V)\twoheadrightarrow{}I_{\breve{\mathcal{F}}}(\mathbf{1})$ in $\mathcal{F}$. Since $I_{\breve{\mathcal{F}}}$ preserves
  the terminal object, this is actually to say $I_{\breve{\mathcal{F}}}(p) :
  I_{\breve{\mathcal{F}}}(V)\twoheadrightarrow{}\mathbf{1}$ in $\breve{\mathcal{F}}$.  $I_{\breve{\mathcal{F}}}(p)$ must be
  a square in $\SET$ of the following kind:
  \[
    \begin{tikzcd}[sep=large,cramped]
      \pi_0(I_{\breve{\mathcal{F}}}(V))
      \arrow[r, two heads]
      \arrow[d]
      &
      \mathbf{1}
      \arrow[d]
      \\
      \Gamma(V)
      \arrow[r]
      &
      \Gamma(\mathbf{1})
    \end{tikzcd}
  \]

  Because the upstairs morphism is a surjection, we know that $\pi_0(I_{\breve{\mathcal{F}}}(V))$ is non-empty;
  therefore, because we have a map $\pi_0(I_{\breve{\mathcal{F}}}(V))\to\Gamma(V)$, we can see that
  $\Gamma(V)$ is non-empty, i.e.\ we have a global section of $\rho :
  \mathbf{1}\to{}V$ in $\mathcal{F}$.
  By precomposition and Kripke-Joyal monotonicity, then, we have a global
  section $\beta\circ\rho:\mathbf{1}\to{}X$ such that
  $\mathbf{1}\Vdash\phi(\beta\circ\rho)$.
\end{proof}

\begin{theorem}[Disjunction property]

  Suppose that $\mathcal{F}\models\phi(\alpha)\lor\psi(\alpha)$ for some
  $\alpha:\mathbf{1}\to{}X$; then either $\mathcal{F}\models\phi(\alpha)$ or
  $\mathcal{F}\models\psi(\alpha)$.

\end{theorem}

\begin{proof}

  We will use essentially the same technique as in our proof of
  Theorem~\ref{thm:existence-property}. Unwinding the Kripke-Joyal semantics of
  the topos, we have morphisms $p:V\to\mathbf{1}$ and $q:W\to\mathbf{1}$ such
  that $[p,q]:V+W\to\mathbf{1}$ is an epimorphism and moreover
  $V\Vdash\phi(\alpha\circ{}p)$ and $W\Vdash\psi(\alpha\circ{}q)$. As above,
  $[p,q]$ lifts to an epimorphism in $\breve{\mathcal{F}}$ as follows:
  \[
    \begin{tikzcd}[sep=large,cramped]
      \pi_0(I_{\breve{\mathcal{F}}}(V+W))=\pi_0(I_{\breve{\mathcal{F}}}(V))+\pi_0(I_{\breve{\mathcal{F}}}(W))
      \arrow[r, two heads]
      \arrow[d]
      &
      \mathbf{1}
      \arrow[d]
      \\
      \Gamma(V+W)=\Gamma(V)+\Gamma(W)
      \arrow[r]
      &
      \mathbf{1}
    \end{tikzcd}
  \]

  Note that the global sections functor for $\mathcal{F}$ preserves finite
  colimits, and $\pi_0$ preserves all
  colimits~\citep[Proposition~7.7.1(l)]{taylor:1999}. Because the upstairs
  morphism is a surjection, we know that $\pi_0(I_{\breve{\mathcal{F}}}(V))+\pi_0(I_{\breve{\mathcal{F}}}(W))$ is
  non-empty, whence we must have either a global section $r\in\Gamma(V)$ or a
  global section $s\in\Gamma(W)$.

  Supposing we have a global section $r:\mathbf{1}\to V$ in $\mathcal{F}$, then
  by Kripke monotonicity, we have
  $\mathbf{1}\Vdash\phi(\alpha\circ{}p\circ{}r)$. On the other hand, if we have
  a global section $s:\mathbf{1}\to{}W$, then we would have
  $\mathbf{1}\Vdash\psi(\alpha\circ{}q\circ{}s)$.

\end{proof}

\subsection{Connection with the method of computability}

As we have alluded to in the previous section, the gluing category always
functions as the ``category of suitable logical predicates'', with the meaning
of ``suitable'' negotiated by choice of gluing functor. Most instances of the
logical relations/predicates technique can be phrased as an instance of the
more general gluing construction.

\subsubsection{Proof (ir)relevance}

The native notion of logical ``predicate'' which is induced by the gluing
construction is a proof-relevant one, whereas in the method of computability,
one generally studies predicates in the classical, proof-irrelevant sense. This
restriction is easily accounted for by making a slight adjustment to the
categories involved.

Writing $\mathbb{C}$ for the classifying category of our theory, if we take a
category $\mathcal{E}$ to be our semantic domain, we can form categories of
proof-relevant logical predicates and proof irrelevant logical predicates
respectively along a functor $F:\mathbb{C}\to\mathcal{E}$ as
follows:\footnote{For intuition, consider the specific example where
$\mathcal{E}$ is $\SET$ and $\mathcal{F}$ is the global sections functor, as in
Section~\ref{sec:freyd-cover}.}
\begin{mathpar}
  \begin{tikzcd}[sep=large,cramped]
    {\mathsf{Gl}}
    \arrow[r]
    \arrow[d, -{Triangle[open]}]
    \arrow[dr, phantom, pos = 0, "\lrcorner"]
    &
    \mathcal{E}^{\to}
    \arrow[d,-{Triangle[open]}, "\Cod"]
    \\
    \mathbb{C}
    \arrow[r, swap, "F"]
    &
    \mathcal{E}
  \end{tikzcd}
  \and
  \begin{tikzcd}[sep=small,cramped]
    {\mathsf{Gl}_{\mathit{irr}}}
    \arrow[rr]
    \arrow[dd]
    \arrow[drr, phantom, pos = 0, "\lrcorner"]
    &&
    \mathcal{E}^{\to}_{\mathit{mono}}
    \arrow[d]
    \\
    &&
    \mathcal{E}^{\to}
    \arrow[d,-{Triangle[open]},"\Cod"]
    \\
    \mathbb{C}
    \arrow[rr, swap, "F"]
    &&
    \mathcal{E}
  \end{tikzcd}
\end{mathpar}

When $F$ is the global sections functor (and thence $\mathsf{Gl}$ is the Freyd
cover or the scone of $\mathbb{C}$), $\mathsf{Gl}_{\mathit{irr}}$ is often
called the ``subscone'' of $\mathbb{C}$.

\subsubsection{Relations vs predicates}

What we have seen so far corresponds to the technique of unary logical
relations, but the abstraction scales easily to the case of binary (and
$n$-ary) logical relations by replacing $\mathbb{C}$ with
$\mathbb{C}\times\mathbb{D}$, as described in~\citet{mitchell-scedrov:1993}. To
see the connection with binary logical relations, it will be instructive to
work out explicitly the case for exponentials in the subscone of
$\mathbb{C}\times\mathbb{C}$, which we will write
$\widetilde{\mathbb{C}\times\mathbb{C}}$.

First, observe that the exponential in the product of two cartesian closed
categories is calculated pointwise; so for $(A_0,A_1), (B_0,B_1) :
\mathbb{C}\times\mathbb{C}$, we have ${(B_0,B_1)}^{(A_0,A_1)} = ({B_0}^{A_0},
{B_1}^{A_1})$.

An object in $\widetilde{\mathbb{C}\times\mathbb{C}}$ is a monomorphism
$R\rightarrowtail\Gamma(A,B)$ where $\Gamma$ is the global sections functor for
$\mathbb{C}\times\mathbb{C}$. Because the global sections functor preserves
finite limits, this is to say that we have a monomorphism
$R\rightarrowtail\Gamma(A)\times\Gamma(B)$, in other words \emph{a relation on
the closed terms of type $A$ and $B$ in the language $\mathbb{C}$}.

We wish to inspect for ourselves the exponential object in
$\widetilde{\mathbb{C}\times\mathbb{C}}$. As we saw earlier on, to form the
exponential in the gluing category we first take the following pullback:
\[
  \begin{tikzcd}[sep=large]
    {E}
    \arrow[dr, phantom, pos = 0, "\lrcorner"]
    \arrow[r]
    \arrow[d, tail]
    &
    {S^R}
    \arrow[d, tail]
    \\
    {\Gamma\Parens{{B_0}^{A_0}}\times\Gamma\Parens{{B_1}^{A_1}}}
    \arrow[r]
    &
    \Parens{\Gamma(B_0,B_1)}^R
  \end{tikzcd}
\]

Then, we define the exponential
${\Parens{S\rightarrowtail\Gamma(B_0)\times\Gamma(B_1)}}^{\Parens{R\rightarrowtail\Gamma(A_0)\times\Gamma(A_1)}}$
to be the monomorphism on the left. Now, unfolding definitions, a global element of this
exponential is simply a pair of closed terms $\cdot\vdash{}F_0:A_0\to{}B_0$ and
$\cdot\vdash{}F_1:A_1\to{}B_1$ together with a function $H : S^R$ which is
tracked by $(F_0,F_1)$; unwinding further, this means only that for all
$\cdot\vdash{}a_0:A_0$ and $\cdot\vdash{}a_1:A_1$, if $(a_0,a_1)\in{}R$, then
$(F_0(a_0), F_1(a_1))\in{}S$.

\subsubsection{Kripke/Beth/Grothendieck logical relations}

A common generalization of the method of computability is to use a logical
relation which is indexed in some partial order (or even a category), subject
to a functoriality condition. In the literature, these are called \emph{Kripke
logical relations}, and indeed, the construction that we used to prove
normalization of free $\lambda$-theories in Section~\ref{sec:nbg} is the
proof-relevant unary Kripke instance of the gluing abstraction, where the
worlds are contexts of variables linked by renamings.

Many other variations of indexed logical relations appear in the wild, and
nearly all of these are already accounted for within the abstraction.  For
instance, by imposing Grothendieck topology on the base poset or category and
requiring a \emph{local character} condition in addition to monotonicity, one
can develop something which might be called \emph{Beth/Grothendieck logical
relations} (see~\citet{coquand-mannaa:2016},
\citet{altenkirch-dybjer-hofmann-scott:2001} and \citet{fiore-simpson:1999} for
examples).

\begin{remark}[Terminology]
  In the literature~\citep{jung-tiuryn:1993,fiore-simpson:1999,fiore:2002}, the
  proof irrelevant version of this construction appears under the somewhat
  confusing name ``Kripke Relations of Varying Arity''---confusing because it is
  not immediately clear what it has to do with the arity of a relation.

  In the early literature (such as~\citet{jung-tiuryn:1993}), there was
  some resistance to explaining what these were in a more conceptual way, but
  as described in~\citet{fiore-simpson:1999}, these have a simple
  characterization as \emph{internal relations} of a certain kind within a
  presheaf topos which corresponds exactly to a proof irrelevant version of the
  construction we describe in these notes.
\end{remark}

\begin{example}[Independence of Markov's Principle]

  In~\citet{coquand-mannaa:2016}, the method of computability was used to
  establish the independence of Markov's Principle from Martin-L\"of Type
  Theory using a forcing extension over Cantor space $\mathcal{C}$. We will
  briefly describe how the construction in that paper fits into the framework
  of gluing.

  Letting $\mathbb{C}$ be the classifying category of the forcing extension of
  type theory, we have a fibration $\pi_{\mathcal{C}}:\mathbb{C}\to\mathcal{C}$
  which projects the forcing condition (a representation of compact open in
  Cantor space). Writing $\Sh{\mathcal{C}}$ for the topos of sheaves over
  Cantor space, we evidently have a functor $\TM:\mathbb{C}\times\mathbb{C}\to\Sh{\mathcal{C}}$
  which takes a pair of contexts $(\Delta_0,\Delta_1)$ to the following sheaf:
  \[
    (p : \mathcal{C}) \mapsto
      \SetCompr{
        (\delta_0,\delta_1)
      }{
        p \leq \pi_{\mathcal{C}}(\Delta_0)
        \land
        p \leq \pi_{\mathcal{C}}(\Delta_1)
        \land
        {\cdot\mathrel{\vdash_p}\delta_0:\Parens{\Delta_0}_{\vert{}p}}
        \land
        {\cdot\mathrel{\vdash_p}\delta_1:\Parens{\Delta_1}_{\vert{}p}}
      }
  \]

  (The above is a sheaf, because the topology on $\mathcal{C}$ is subcanonical,
  and because the calculus contains a rule for local character.)

  Now, consider the gluing category obtained from the following pullback:
  \[
    \begin{tikzcd}[sep=large,cramped]
      {\mathsf{Gl}}
      \arrow[r]
      \arrow[dd]
      \arrow[ddr, phantom, pos = 0, "\lrcorner"]
      &
      \Sh{\mathcal{C}}^{\to}_{\mathit{mono}}
      \arrow[d]
      \\
      &
      \Sh{\mathcal{C}}^{\to}
      \arrow[d,-{Triangle[open]}, "\Cod"]
      \\
      \mathbb{C}\times\mathbb{C}
      \arrow[r, swap, "F"]
      &
      \Sh{\mathcal{C}}
    \end{tikzcd}
  \]

  Viewed externally, the objects of $\mathsf{Gl}$ are $\mathcal{C}$-indexed
  binary relations on closed terms in $\mathbb{C}$ which enjoy both
  monotonicity and local character. By examining the cartesian closed structure
  of $\mathsf{Gl}$, it can be seen (as above) that the logical relations for
  each connective match the na\"ive ones.

\end{example}

%
%
%
%
%
%

\section*{Acknowledgments}

We are thankful to Jonas Frey for explaining the Freyd cover construction for
the free topos; and to Marcelo Fiore, Michael Shulman and Darin Morrison for
helpful conversations about conceptual proofs of normalization. Thanks also to
Daniel Gratzer and Lars Birkedal for their comments on a draft of this note.

The authors gratefully acknowledge the support of the Air Force Office of
Scientific Research through MURI grant FA9550-15-1-0053 and the AFOSR project
`Homotopy Type Theory and Probabilistic Computation', 12595060. Any opinions,
findings and conclusions.  Any opinions, findings and conclusions or
recommendations expressed in this material are those of the authors and do not
necessarily reflect the views of the AFOSR.%

\nocite{coquand-dybjer:1997, shulman:2015, crole:1993, dagand-scherer:2015, nlab:free-topos, nlab:freyd-cover}
\nocite{clairambault-dybjer:2015}
\bibliographystyle{plainnat}
\bibliography{references/refs}

\end{document}